\documentclass{vldb}

\usepackage{graphicx}
\usepackage{balance}  

\usepackage{graphicx}
\usepackage{float}
\usepackage{amsmath}

\usepackage{centernot}
\usepackage{algorithmic}
\usepackage{algorithm}
\usepackage{array}
\usepackage{colortbl}
\usepackage[makeroom]{cancel}
\usepackage{url}

\newdef{assumption}{Assumption}
\newdef{example}{Example}
\newdef{definition}{Definition}
\newdef{theorem}{Proposition}
\newdef{proposition}{Proposition}

\begin{document}




\title{Scalable Entity Resolution Using Probabilistic Signatures on
       Parallel Databases}




\numberofauthors{3} 

\author{
\alignauthor Yuhang Zhang
  \titlenote{Corresponding author: yuhang.zhang@austrac.gov.au \\
             \copyright \, Commonwealth of Australia}\\
       \affaddr{\small AUSTRAC} 
\alignauthor Kee Siong Ng\\
       \affaddr{\small AUSTRAC / ANU} 
\alignauthor Michael Walker \\
        \affaddr{\small AUSTRAC} 
\and \alignauthor Pauline Chou\\
     \affaddr{\small AUSTRAC}
  \alignauthor Tania Churchill \\
  \affaddr{\small AUSTRAC} 
\alignauthor Peter Christen\\
  \affaddr{\small Australian National University} 
}


\maketitle

\begin{abstract}
Accurate and efficient entity resolution is an open challenge of
particular relevance to intelligence organisations that collect large datasets from disparate sources with differing levels of quality and
standard.
Starting from a first-principles formulation of entity resolution,  
this paper presents a novel Entity Resolution algorithm that introduces a data-driven blocking and record linkage technique based on the probabilistic identification of entity signatures in data.
The scalability and accuracy of the proposed algorithm are evaluated using benchmark datasets and shown to achieve state-of-the-art results.
The proposed algorithm can be implemented simply on modern parallel databases, which allows it to be deployed with relative ease in large industrial applications. 
\end{abstract}


\section{Introduction}
  
Entity resolution (ER) is the process of identifying records that refer
to the same real-world entity. 
Accurate and efficient ER is needed in various data-intensive
applications, including but not limited to health studies, fraud
detection, and national censuses~\cite{Christen:2012:DMC:2344108}.
More specifically, ER plays a pivotal role in the context of
Australia's whole-of-government approach to tackle our most pressing
social issues -- including terrorism and welfare fraud -- by combining
and analysing datasets from multiple government agencies. 

Two typical challenges in entity resolution are imperfect
data quality and large data size.
Common data quality issues that can introduce ambiguity in the ER process include:
\begin{itemize}\itemsep1.5mm\parskip0mm
  \item \textbf{Incompleteness}: Records with incomplete attribute
        values or even missing attribute values.
  \item \textbf{Incompatible formats}: The formats of names,
        addresses, dates, numbers, etc., can be different between
        countries, regions, and languages.
  \item \textbf{Errors}: Records containing wrong information due to
        either user or system errors, or deliberate attempts at
        obfuscation are widely seen in databases.
  \item \textbf{Timeliness}: Records have become outdated due to poor maintenance or
        data refresh practices, such as people changing their name or
        address.
\end{itemize}

In databases containing upwards of tens to hundreds of millions records, ER can also be challenging because exhaustively comparing records in a
pairwise manner is computationally infeasible~\cite{Chr12b}. 
In fact, any algorithm with time complexity worse than
linear is prohibitive on large databases.

In this paper, we present a simple and scalable ER algorithm that addresses the challenges of performing ER on poor quality and high volume data. 
The key ideas behind our proposed approach are described next. 


\subsubsection*{Using Redundancy to Overcome Data Quality Issues}
The most common way to tackle data quality issues is to standardise
and cleanse raw data before the linking operation
\cite{Christen:2012:DMC:2344108}. Standardisation and cleansing are
umbrella terms covering operations which can fill in incomplete data,
unify inconsistent formats, and remove errors in data.

The problem with standardisation and cleansing is that it is in
itself a challenging problem. For example, \emph{01/02/2000} can be
parsed as either \emph{1st of Feb 2000} or \emph{2nd of Jan 2000}.
\emph{St} can mean either \emph{Street} or \emph{Saint} in addresses.
If a mistake is made during standardisation and cleansing, it is
usually difficult to recover from it to perform linkage correctly.

Instead of standardising and cleansing data into canonical forms, we
rely on redundancy in data to overcome data quality issues. We say a
record contains redundancy if one of its subrecords can uniquely
identify the same entity. For example, if there is only one
\emph{John Smith} living in \emph{Elizabeth Street}, then \emph{John
Smith, 45 Elizabeth Street} as a record of a person contains redundancy,
because specifying street number \emph{45} is not really necessary.

Redundancy exists widely in data. Not every country has a city named
\emph{Canberra}. Not every bank has a branch in \emph{Bungendore}. As
an extreme case, three numbers \emph{23 24 5600} can be sufficient
to specify an address globally, if there is only one address in the
world containing these three numbers at the same time. In this case,
we do not even need to know if \emph{23} is a unit number or the first
part of a street number. 
Such seemingly extreme examples are actually quite common in practice. 
For example, $1,374,998$ of the $13.9$ million Australian addresses in the Open Address
\cite{openaddress} database can be uniquely identified by just three numbers in them.

Redundancy simplifies ER. If two records share a common subrecord that
can be used to uniquely identify an entity, then these two
records can be linked no matter what data quality issues they each have. 
We call such a subrecord a \emph{signature} of its entity. 
Probabilistic identification of signatures in data and linking records using such probabilistic signatures is the first key idea of our algorithm. 


\subsubsection*{Data-Driven Blocking using Signatures}
Blocking is a widely used technique to improve ER
efficiency~\cite{Christen:2012:DMC:2344108}. Na{\"i}vely, linking two
databases containing $m$ and $n$ records respectively requires
$O(mn)$ record pair comparisons. Most of these comparisons lead to
non-matches, i.e.\ they correspond to two records that refer to
different entities. To reject these non-matches with a lower cost,
one may first partition the raw records according to criteria
selected by a user. These criteria are called \emph{blocking
keys}~\cite{Chr12b}. Examples of blocking keys include attributes
such as first and last name, postcode, and so on. During linkage,
comparison is only carried out between records that fall into the
same partition, based on the assumption that records sharing no
blocking keys do not match with each other.

The efficiency and completeness of ER is largely determined by
blocking-key selection, which again is challenging in itself. If the
keys are not distinctive between disparate entities, many irrelevant
records will be placed into the same block, which gains little
improvement in efficiency. If the keys are not invariant with
respect to records of the same entities, records of the same entity
will be inserted into different blocks and many true matching record
pairs will be missed. If the key values do not distribute evenly
among the records, the largest few blocks will form the bottleneck of
ER efficiency. When dealing with a large dataset, it is challenging
to balance all these concerns. Moreover, the performance of blocking
keys also depends on the accuracy of any data standardisation and
cleansing performed~\cite{Christen:2012:DMC:2344108}.

In an ideal world, we would like to use signatures as the blocking key and
place only records of the same entity into the same block. 
In practice, we do not know which subrecords are signatures but we can still approximate the strategy by blocking on probabilistically identified signatures, as we describe in
Section~\ref{sec-signatures}.
These probabilistic signatures tend to be empirically distinctive and exhibit low-frequency in the database, which allows small and accurate blocks to be constructed.
The only risk is these blocking keys may not be invariant with respect to records of the same entities. 
To address this, we introduce an inter-block connected component algorithm,
which is explained next.

\subsubsection*{Connected Components for Transitive Linkage}
As discussed above, the blocking-by-probabilistic-signature technique
leads to quite targetted blocking of records, with high precision but
possibly low recall. This is in contrast to standard blocking
techniques that tend to have low precision but high
recall~\cite{Chr12b}. 
To compensate for the loss in recall, we allow each record to be inserted
into multiple blocks, using the fact that each record may contain
multiple distinct signatures. Moreover, to link records of the same
entity that do not share the same signature, we allow two records in
two different blocks to be linked if they are linked to the same
third record in their own blocks.
To implement such an indirect (transitive) link, we run a connected
component algorithm to assign records connected directly or
indirectly with the same label (entity identifier). 

A particular challenge in our context is the size of the graphs we
have to deal with. There are as many nodes as the number of records.
Such a graph can be too large to fit into main memory. Random
access to nodes in the graph, which is required by traditional
depth/breadth-first search algorithms, might therefore not be
feasible.
To addres this, we propose a connected-component labelling algorithm that fits large
graphs that are stored in a distributed database. The algorithm uses
standard relational database operations, such as grouping and join,
in an iterative way and converges in linear time. This connected
component operation allows us not only to use small-sized data
blocks, but also to link highly inconsistent records of the same
entity transitively.


\subsubsection*{Implementation on Parallel Databases}
Massively parallel processing databases like Teradata and Greenplum
have long supported parallelised SQL that scales to large datasets.
Recent advances in large-scale in-database analytics platforms
\cite{DBLP:journals/pvldb/HellersteinRSWFGNWFLK12,
Zaharia:2010:SCC:1863103.1863113} have shown how sophisticated
machine learning algorithms can be implemented on top of a
declarative language like SQL or MapReduce to scale to
Petabyte-sized datasets on cluster computing.

One merit of our proposed method is it can be implemented on
parallelised SQL using around ten SQL statements. As our experiments
presented in Section~\ref{sec-experiments} show, our algorithm can
link datasets containing thousands of records in seconds, millions of
records in minutes, and billions of records in hours on medium-sized clusters built using inexpensive commodity hardware.

\subsubsection*{Paper Contributions}
The contributions of this paper is a novel ER algorithm that
\begin{enumerate}\itemsep1mm\parskip0mm
  \item introduces a probabilistic technique to identify, from
        unlabelled data, entity signatures derived from a first-principles 
        formulation of the ER problem;

 \item introduces a new and effective data-driven blocking technique based on the occurrence of common probabilistic signatures in two records;
         
  \item incorporates a scalable connected-component labelling
        algorithm that uses inverted-index data structures and
        parallel databases to compute transitive linkages in large
        graphs (tens to hundreds of millions of nodes);
        
  \item is simple and scalable, allowing the whole algorithm to be written in
        $\sim$10 standard SQL statements on
        modern parallel data platforms like Greenplum and Spark;      
        
  \item achieves state-of-the-art performance on several benchmark
        datasets and pushes the scalability boundary of existing ER algorithms.
\end{enumerate}
Our paper also provides a positive answer to an open research problem
raised by \cite{papadakis17} about the
existence of scalable and accurate data-driven blocking algorithms.

The paper is organised as follows. 
In Section~\ref{sec:problem formulation} we formulate the ER problem precisely.
In Section~\ref{sec-signatures} we describe how entity signatures can be identified in a probabilistic way.
In Section~\ref{sec:cc} we propose a scalable graph-labelling algorithm for identifying transitive links. 
We present the overall algorithm for signature-based ER in Section~\ref{sec:signature-er}.
Experimental results are presented in Section~\ref{sec-experiments}, followed by a literature review and discussion in Section~\ref{sec:rel-work} and conclusion in Section~\ref{sec:conclusion}.



\section{Problem Formulation}\label{sec:problem formulation}

The ER problem is usually loosely defined as the problem of
determining which records in a database refer to the same entities.
This informal definition can hide many assumptions, especially on the
meaning of the term ``same entities''. To avoid confusion, we now
define our ER setting in a more precise manner.

\begin{definition}
A {\bf possible world} is a tuple $(W, R, E, D)$, where $W$ denotes a
set of words; $R$ denotes the set of all records, where a record $r
\in R$ is a sequence of words from $W$ (i.e.\ order matters);
$E=\{e_1, e_2, \dots\}$ denotes a set of entity identifiers; and
$D: E \times R$ is a subset of the Cartesian product between $E$ and
$R$.
\end{definition}

We say record $r\in R$ refers to entity $e \in E$, if
$(e,r) \in D$. Note that an entity may be referred to by multiple (possibly
inconsistent) records, and each record may refer to multiple
entities, i.e., there are ambiguous records. Some records may belong
to no entities in $E$. For example, \emph{John Smith, Sydney} is
likely referring to several individuals named \emph{John Smith} who
live in \emph{Sydney}, and therefore this record is ambiguous as it
can refer to any of them. On the other hand, in real-world databases
there are often records that contain randomly generated, faked, or
corrupted values, such as those used to test a system or that were
intentionally modified (for example \emph{John Doe} or
\emph{(123) 456-7890}) by a user who does not want to provide their
actual personal details~\cite{Chr16}.

In practice, a possible world is only `knowable' through a (finite)
set of observations sampled from it.

\begin{definition}\label{def:observations}
Given a possible world $(W,R,E,D)$, we can sample an $(e,r)$ pair
using some (usually unknown) probability distribution on $D$. By
repeating the sampling $n$ times, we obtain a set of {\bf labelled
observations} of the possible world, $\{ (e_i, r_i) \}_{i=1 \ldots
n}$. From labelled observations, we can derive {\bf unlabelled
observations} by removing all the $e_i$'s.
\end{definition} 

Roughly speaking, ER is the problem of reconstructing labelled
observations from unlabelled observations.

\begin{definition}\label{def:er}
Given a set of unlabelled observations $O$
sampled from a possible world $(W,R,E,D)$, {\bf entity resolution}
is the problem of constructing a partition of $O$
\[ O = \bigcup_{i} O_i \] satisfying the following two properties:
(1) for each $O_i$, there exists an $e \in E$ such that
$\{ (e,r) \,|\, r \in O_i \} \subseteq D$; and
(2) the number of partitions is minimised.
\end{definition}

A trivial way to satisfy the first condition of Definition~\ref{def:er} is to assign each record in $O$ to its own partition. 
The second condition is needed to make sure records of the same underlying entity are assigned to the same partition.
ER as defined above is an underconstrained optimisation problem. For example,
there could be multiple ways of partitioning a set of unlabelled
observations that all satisfy Definition~\ref{def:er} because of the
existence of ambiguous records that refer to multiple entities.
We need further assumptions on the structure of possible worlds, in
particular the structure of $D$, to be able to distinguish between
possible solutions.
The following are some common ways of refining the ER problem, each
with its own assumptions on $D$.
\begin{enumerate}
  \item \textbf{Supervised Learning Methods}:
        The first class of methods assume that a set of labelled
        observations is available with which we can apply supervised
        learning techniques to label a much larger set of unlabelled
        observations~\cite{Christen:2012:DMC:2344108,
        Kopcke:2010:EER:1920841.1920904}. In particular, these methods
        assume the joint probability distribution of entities and
        records $P:E \times R \rightarrow [0,1]$ induced by the
        unknown $D$ and the observations' sampling process have enough
        structure, in the learning-theoretic sense
        of~\cite{anthony-bartlett99, bartlett-mendelson02}, to be
        learnable from finite sample sizes and suitable model
        classes. Note the probability of learning good models is with
        respect to a probability distribution on the possible worlds
        that are consistent with a set of labelled observations.

  \item \textbf{Distance Based Methods}:
        The second class of methods work only with unlabelled
        observations and assume records can be embedded into a metric
        space, where records of an entity fall in a compact
        region~\cite{Jin03}. One first finds such a metric space in
        the form of a suitable distance function that incorporates
        domain knowledge on what constitutes records that
        likely belong to the same entity. Records are then clustered,
        either exactly through a nearest-neighbour algorithm or
        approximately using blocking or clustering
        techniques~\cite{Chr12b}, and then labelled based on some
        linkage rule. This is by far the most common approach to ER
        and has a long history going back nearly fifty
        years~\cite{fellegi-sunter69}. Distance based methods are
        sometimes used in conjunction with supervised learning
        algorithms to determine the linkage rule or clustering
        thresholds~\cite{Christen:2012:DMC:2344108}.
\end{enumerate}

\subsubsection*{Signature-Based Entity Resolution}
We consider in this paper a family of signature-based methods, where
we assume each entity has distinctive signatures that can be detected
from a set of unlabelled observations (sampled from a possible world)
and that the signatures so-detected can be used to link records of the
same entities. Compared to the other two types of methods described
above, signature-based methods make a number of alternative
assumptions on the structure of possible worlds which we now describe.

A sufficient condition for a record to be a signature is that it belongs to one and only one entity in a possible world. 
However, the condition is not a necessary one because a signature of an entity $e$ does not have to be a record of $e$, but merely one computationally derivable from a record belonging to $e$. 
We now formalise the idea.

\begin{definition}\label{def:sig3}
Given a possible world $(W,R,E,D)$ and a computable relation $T: R\times R$,
record $s$ is a \textbf{signature} of an entity $e$ subject to $T$ iff 
\begin{enumerate}
  \item $\exists r \in R$ such that $(s,r) \in T$ and $(e,r)
        \in D$; and
  \item $\forall f \in E$, $\forall r \in R$, if $(f,r) \in D$ and $(s,r) \in T$, then $e = f$. \label{def:sig3:part2}
\end{enumerate}
\end{definition}
One way to understand Definition~\ref{def:sig3} is that $T$ defines a computable
transform of a record $s$ into all its variants $\{ r \,|\, (s,r) \in T \}$,
and $s$ is a signature of $e$ if all its variants obtained via $T$ 
contain and only contain records belonging to $e$. 
A signature provides sufficient condition to link two records.
\begin{proposition}\label{prop:sufficient}
Let $P$ be a possible world, $T$ a relation, and $s$ a signature of an entity subject to $T$. 
Two unlabelled observations $r,t$ sampled from $P$ belong to the same entity if $(s,r) \in T$ and $(s,t) \in T$.
\end{proposition}
\begin{proof}
By Definition~\ref{def:observations}, there exists entities $e_r, e_t \in E$ such that $(e_r,r) \in D$ and $(e_t,t) \in D$. 
By Definition~\ref{def:sig3} part~\ref{def:sig3:part2}, we can infer $e_r = e$ from $(s,r) \in T$ and $(e_r,r) \in D$, and $e_t = e$ from $(s,t) \in T$ and $(e_t,t) \in D$.
\end{proof}

\begin{figure}[t!]
	\centering
    \includegraphics[width=0.47\textwidth]{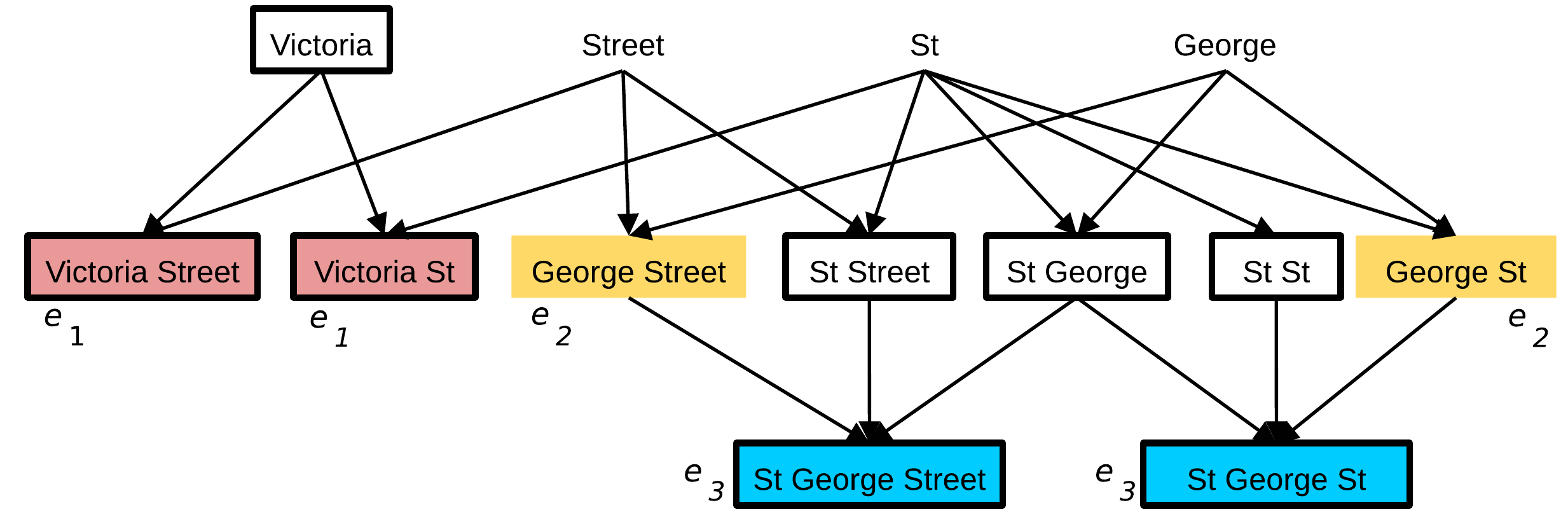}
	\caption{A possible world of addresses of three entities ($e_1, e_2, e_3$) and their signatures as described in Example~\ref{ex:subrecord}, where
    records of different entities are shown in different colours, and
    thick outlines show records/subrecords which are signatures
    subject to the subrecord relation. \label{fig:st} }
\end{figure}

To familiarise readers with our formulation, we now describe some traditional ER algorithms with our concepts.

\begin{example}
Rule-based ER: link two records if they share patterns
predefined by some regular expressions, i.e., link $r$ and $s$ if
$\mathtt{regexp}(r) = \mathtt{regexp}(s)$:
  \[ T=\{(s,r) \mid \mathtt{regexp}(s) = \mathtt{regexp}(r)\} \]
\end{example}
\begin{example}
Distance-based ER: link two records if their distance is
below a threshold $\tau$ according to a selected/learned distance
function $d$~\cite{Christen:2012:DMC:2344108}:
  \[ T=\{(s,r) \mid d(s,r)<\tau\} \]
\end{example}


A common design in traditional ER algorithms is to find a relation $T$ which contains all pairs of records $s$ and $r$ referring to the same entities. Two records $s$ and $r$ are then linked using the fact $(s,s) \in T$, $(s,r) \in T$ and Proposition~\ref{prop:sufficient}. The concept of signature is not explicitly used in this design because every unambiguous record in a dataset will then be a signature. The challenge is all about finding the relation $T$.

In this paper, we follow a different strategy. Instead of searching
for an unknown relation $T$, we start with one (or more) known relation(s) $T$ and then search for records which are signatures subject to this known $T$

A trivial example is when $T$ is defined by equality: $T=\{(s,r) \mid s=r\}$.
Signatures subject to equality are those records that belong to one and only one entity.
These signatures are not particularly interesting, as they can only be used to find exact duplicate records in a database.

\subsubsection*{Signatures based on the Subrecord Relation}
Consider now the more powerful $T$ defined by the \emph{subrecord} relation.
Given a record $r$, we say $s$ is a subrecord of $r$, denoted
$s \preceq r$, if $s$ is a subsequence of $r$, i.e.\ $s$ can be
derived from $r$ by deleting some words without changing the order of
the remaining words.
Equivalently, we sometimes say $r$ is a superrecord of $s$ to mean $s \preceq r$.

\begin{example}\label{ex:subrecord}
Define $T=\{(s,r) \mid s \preceq r\}$. Suppose we have the
possible world shown in Figure~\ref{fig:st}, in which
\begin{itemize}
	\item $W$=\{Victoria, Street, St, George\}
	\item $E$=\{$e_1$, $e_2$, $e_3$\}
	\item $D$=\{($e_1$, ``Victoria Street''), \-\hspace{8pt}($e_1$,
          ``Victoria St''), \\ 
	      \-\hspace{17pt} ($e_2$, ``George  Street''),
          \-\hspace{12pt}($e_2$, ``George St''), \\
          \-\hspace{17pt} ($e_3$, ``St George Street''), ($e_3$,
          ``St George St'')\}.
\end{itemize}
Figure~\ref{fig:st} shows the six records in $D$ as well as their
subrecords. Records of different entities are shown in different
colours. We add thick outlines to records/subrecords which are
signatures subject to the subrecord relation. For example, the word
\emph{Victoria} is a signature because all records in $D$ containing
\emph{Victoria} as a subrecord belong to the same entity $e_1$. We
can therefore link these records during ER despite their inconsistency. In contrast,
\emph{Street} is not a signature because it has three superrecords in
$D$ that belong to three distinct entities. Since a record is a
subrecord of itself, some of the records appearing in $D$ are
signatures as well. A special case is entity $e_2$, which does not
have any signature subject to the subrecord relation because all
its records,\emph{George Street} and \emph{George St}, are
subrecords of another entity's records as well. Therefore all their
subrecords are shared by at least two entities. However, entities
like this, whose records are all subrecords of other entities, are
rare in practice, especially when multiple attributes are considered.
\end{example}

From the example above, we can also see the following distinction between our method and traditional ER methods. By explicitly introducing the concept of signatures, we no longer deal with pairwise linkage between records in $O$, but the linkage between records in $O$ and signatures.

\begin{figure}[t!]
	\centering
    \includegraphics[width=0.45\textwidth]{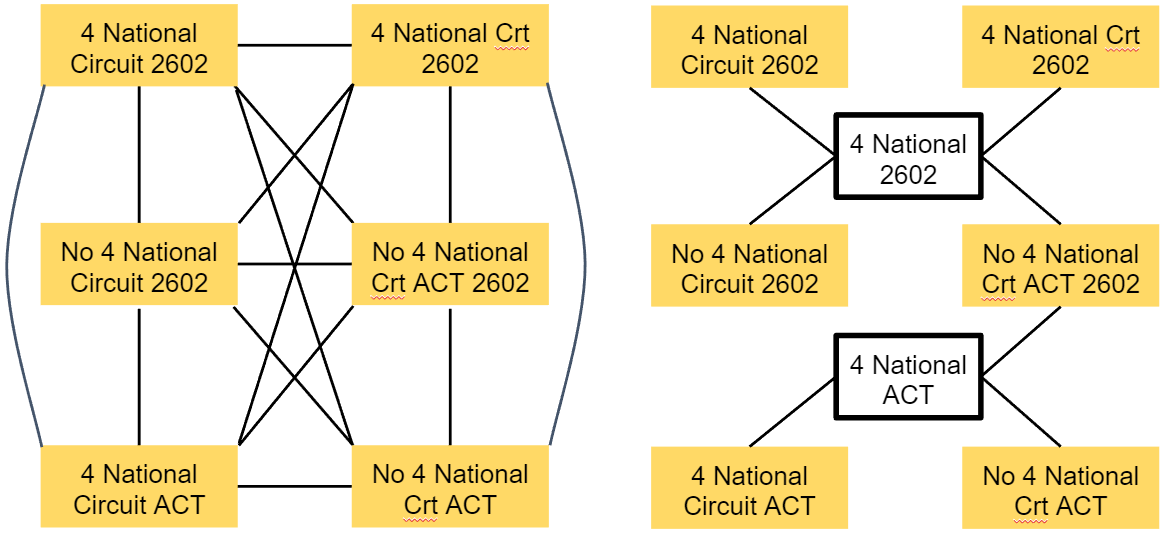}
		\caption{\label{fig:distinction} Left: linkage to be established by traditional ER methods; right: linkage to be established by the proposed method.}
\end{figure}

This distinction is illustrated in Figure~\ref{fig:distinction}, where records are variants of the same address.
Although both graphs depict the same linkage solution, the one used by our method (right-hand side) contains less links due to the usage of signatures. This distinction partly explains why the proposed method is more efficient.

\begin{definition}\label{def:tsubseteq}
Given a set $O$ of unlabelled observations, 
we define $\sqsubseteq$ to be the restriction of $\preceq$ to only terms that are subrecords of observations in $O$.
\[ \sqsubseteq \,= \, \{ (s,r) \mid \exists o\in O.(s \preceq r \wedge r \preceq o) \} \]
\end{definition}

\begin{definition}
Given a set $O$ of unlabelled observations, we call a signature subject to $\sqsubseteq$ a $\sqsubseteq$-signature. 
\end{definition}

\begin{proposition}\label{prop:superrecord signature}
Given a set $O$ of unlabelled observations,
if $s$ is a $\sqsubseteq$-signature of an entity $e$ and 
$s \sqsubseteq r$, then $r$ is also a $\sqsubseteq$-signature of $e$.
\end{proposition}
\begin{proof}
Part 1 of Definition~\ref{def:sig3}: 
Since $s \sqsubseteq r$, there exists $o \in O$ such that $s \preceq r$ and $r \preceq o$.
We have $s \sqsubseteq o$ since $\preceq$ is transitive and $o \in O$. 
To show $(e,o) \in D$, observe that $o \in O$ implies there exists $f$ such that $(f,o) \in D$. 
Since $s$ is a $\sqsubseteq$-signature of e and $s \sqsubseteq r$, we have $e = f$. 

Part 2 of Definition~\ref{def:sig3}:
Consider any $t$ and $f$ such that $r \sqsubseteq t$ and $(f,t) \in D$.
We have $s \sqsubseteq t$ since $\sqsubseteq$ is transitive and $s \sqsubseteq r$. 
Since $s$ is a $\sqsubseteq$-signature of $e$,
we have $f=e$. 
\end{proof}

In practical applications of ER, $\sqsubseteq$-signatures are common.
For example, in a database where entities have unique identifiers such as passport numbers, driver's licenses or tax file numbers,
each unique ID is a signature of its entity. (Recall that the
subrecord relation captures the equality relation as a subset.) Even
in the absence of such unique IDs, countries like Australia have
identity verification systems like the 100 point check \cite{ftract} that allows a
combination of possibly non-unique attributes to be used as a person's
signature.

Given a set of unlabelled observations sampled from an unknown possible world, in the following section we provide an algorithm that
can resolve, with high probability, those entities
that have (one or more) $\sqsubseteq$-signatures.
In the rest of this paper, signatures always refer to
$\sqsubseteq$-signature unless otherwise indicated.


\section{Identification of Signatures}
\label{sec-signatures}

Our general strategy for ER is to {\em probabilistically} identify
signatures from unlabelled observations and then transitively link
records via the identified signatures.

Given a set of unlabelled observations $O$, our first step is to
remove all exact duplicate records to arrive at a deduplicated set
of records. In a deduplicated dataset containing $n$ records, a
subrecord recurs $m$ times if $m$ out of the $n$ records are its
superrecord. By definition, a signature is unique to an entity.
Further, a signature may not appear in every record of the entity to
which it belongs. A non-signature, in contrast, can appear in many
distinct records of many distinct entities. Thus as more and more records of are added to a dataset, after deduplication, the
recurrence frequency of a signature is upper bounded by the number of
distinct records of its entity. The recurrence frequency of a
non-signature, however, may keep on growing. 

This is intuitively
clear from Figure~\ref{fig:st}, where the recurrence frequencies of
non-signature records like \emph{Street} and \emph{St} increases much more quickly, upper-bounded only by the size of the database, as more street names are added into the database.
This difference in recurrence frequency between signatures and non-signatures is the major clue behind our technique to (probablistically) separate them.


\subsection{Probability of Observing a Signature}

Empirically, setting the probability of a subrecord being a signature to go down as its recurrence goes up using
a Poisson distribution with a low mean or a power-law distribution appears sufficient.
In the following, we attempt to derive such a distribution from first principles, which at least will provide an understanding of the inherent assumptions we are making in using such a distribution.

Given a record $s$, the probability of a randomly sampled record $r$ satisfying $s \sqsubseteq r$ is given by a Bernoulli distribution with parameter $p_s$. 
The probability for the given $s$ to recur $k$ times as a subrecord in a deduplicated dataset of size $n$ is therefore governed by a Binomial distribution $\mathit{binom(k;n,p_s)}$.
Now consider the probability of a randomly sampled subrecord to recur $k$ times in a deduplicated dataset of size $n$, which is given by 
\begin{equation}
P(k) = \sum_{s} P(t) \cdot \mathit{binom(k; n, p_s)}.
\end{equation}
If the $p_s$'s are mostly small, which is true in our case, then one can show, from empirical simulations, that the pointwise addition of Binomial distributions with different parameters can be approximated by a Poisson distribution 
\begin{equation}
P(k)\approx e^{-\lambda}\frac{\lambda^k}{k!}
\end{equation}
for a suitable $\lambda$ that can be estimated from data.
Therefore, the recurrence of a subrecord, whether a signature or not, follows Poisson distributions. 
The difference between signatures and
non-signatures is with the average recurrence frequency. 

Denote the set of signatures with $S$. Let $\lambda$ and $\mu$ be
the expected recurrence frequency of a signature and a non-signature,
respectively. The probability of observing a signature or a
non-signature $k$ times is therefore
\begin{equation}
P(k \mid r\in S)=e^{-\lambda}\frac{\lambda^k}{k!} ~\;\mathtt{and}\;~ P(k\mid r\notin S)=e^{-\mu}\frac{\mu^k}{k!}.
\end{equation}

By Bayes rule, when we observe a subrecord $k$ times in a dataset,
the probability for this subrecord to be a signature is given by
\begin{equation}\label{eq:poisson}
P(r\in S\mid k)=\frac{P(k \mid r\in S)P(r\in S)}{P(k)},
\end{equation}
where 
\begin{align}
P\big(k\big) =&P\big(k\wedge r\in S\big)+P\big(k\wedge r\notin S\big)\\
=&P(k \mid r\in S)P(r\in S)+P(k \mid r\notin S)P(r\notin S).
\end{align}
We also assume $P(r \in S)$ 
follows a Bernoulli distribution with parameter $c$:
\begin{equation}
P(r\in S)=c.
\end{equation}
Substituting these into Equation~(\ref{eq:poisson}), we have
\begin{align}
&~P(r\in S\mid k)\\
=&~\frac{P(k \mid r\in S)P(r\in S)}{P(k \mid r\in
S)P(r\in S)+P(k \mid r\notin S)P(r\notin S)}\\
=&~\frac{1}{1+\frac{P(k \mid r\notin S)(1-P(r\in S))}{P(k \mid r\in
S)P(r\in S)}}\\
=&~\frac{1}{1+e^{\lambda-\mu}(\frac{\mu}{\lambda})^k\frac{1-c}{c}}
\end{align}
Letting $a=\frac{\mu}{\lambda}$ and $b=e^{\lambda-\mu}\frac{1-c}{c}$, 
we can state the result as
 \begin{equation}\label{eq:result}
 P(r\in S\mid k)=\frac{1}{1+a^kb}\quad.
 \end{equation}

\begin{figure}[t!]
	\centering
	\includegraphics[width=0.45\textwidth]{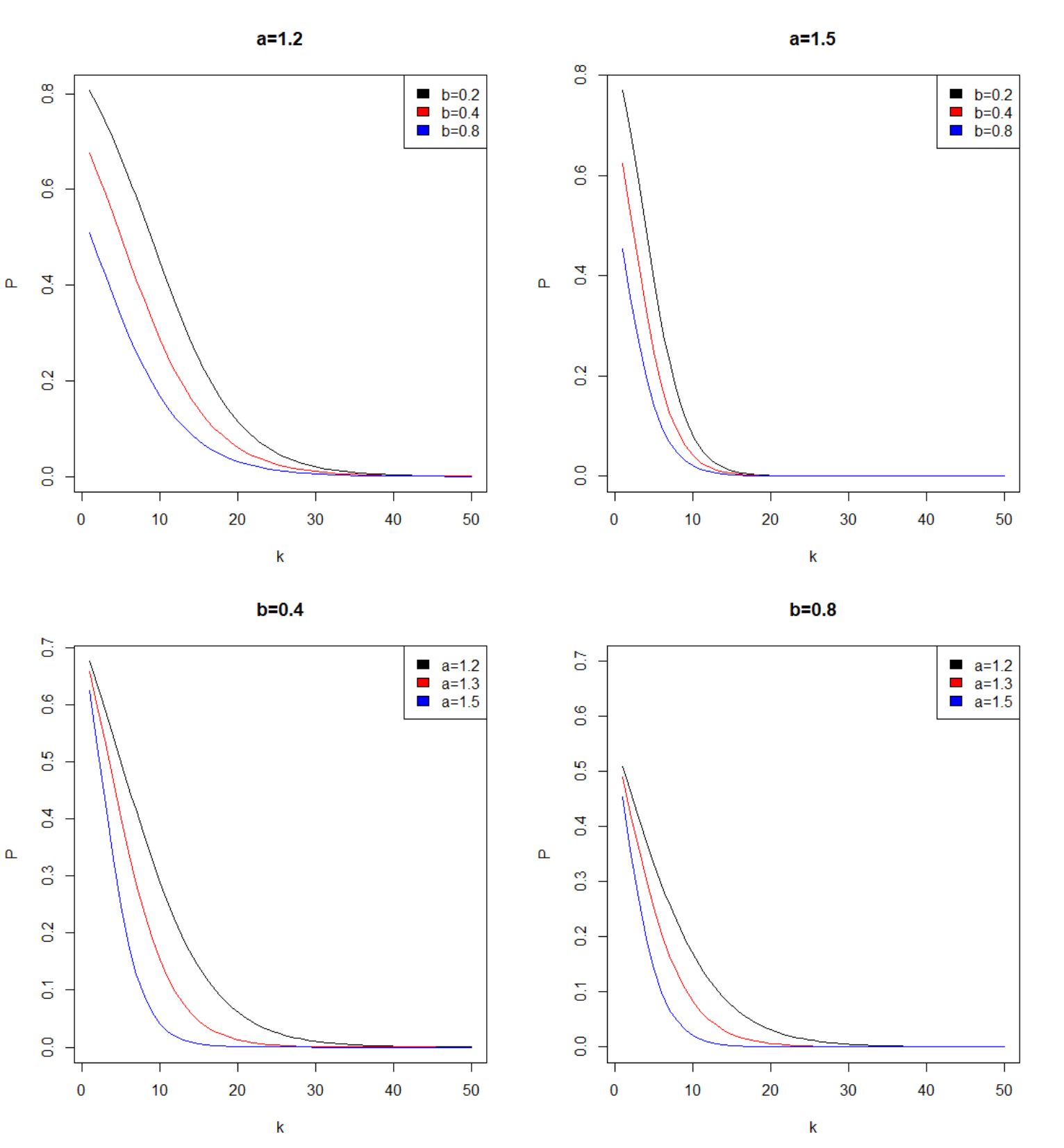}
\caption{Plot of Equation~(\ref{eq:result})\label{fig:ab} by fixing either $a$ or $b$ and changing the other.}
\end{figure}

In practice, since there are more distinct signatures than
non-signatures, i.e.\ $c>1-c$, and a non-signature appears more
frequently than a signature, i.e.\ $\mu>\lambda$, we usually have
$b<1<a$. We can understand the parameters of $P(r \in S \mid k)$
by noting that $a$ controls how fast $P(r \in S \mid k)$ decays as
$k$ increases, and $b$ controls the maximum of $P(r \in S \mid k)$,
as shown in Figure~\ref{fig:ab}.

\subsection{Record Linkage via Common Signatures}
\label{subsec:record linkage}

So far we have worked out how to compute the probability for a single
record to be a signature given its recurrence. In practice, computing
the common subrecords between every pair of records, checking the
recurrence of these subrecords in the database, and then computing
the signature probabilities is prohibitively expensive.

We now show how these probabilities can be approximated efficiently
in a large database. The main idea is to pre-compute a set of
subrecords -- call them \emph{candidate signatures} --
from each record in the database, as well as the probability for each
of these subrecords to be a signature. Given two records $r_i$ and $r_j$, we
approximate the probability for them to share a signature with the
probability of at least one candidate signature shared by both
records being a signature. This approximation can be accelerated by
inverted indices.

More specifically, let $I=\{(s,R_s,p_s)\}$ denote the inverted index
of a database, where each $s$ (inverted index key) denotes a subrecord, $R_s$ denotes the
set of records that contain $s$ as a subrecord, and $p_s=P(s\in S
\mid k=|R_s|)$ is the probability of $s$ being a signature. Computing
linkage probabilities consists of the following steps:
\begin{enumerate}
  \item \textbf{Generation}: From each $(s,R_s,p_s)\in I$, generate
        all tuples of the form of $(r_i,r_j,s,p_s)$ such that
        $r_i,r_j\in R_s$.
  \item \textbf{Elimination}: From all tuples $(r_i,r_j,s,p_s)$
        containing the same $r_i$ and $r_j$, we eliminate those
        tuples whose $s$ appears as a subrecord in another tuple.
        Following Proposition~\ref{prop:superrecord signature}, this is because if a subrecord is a
        signature, then all its superrecords must be signatures.
        We therefore only need to assess the superrecords.
  \item \textbf{Product}: We assume the probability for two
        subrecords being signatures to be independent if they are
        not a subrecord of each other. The probability of $r_i$ and
        $r_j$ sharing a signature can then be computed as
        $1-\prod_s(1-p_s)$ over all $s$ in the remaining tuples
        $(r_i,r_j,s,p_s)$ for the record pair $r_i$ and $r_j$.
\end{enumerate}

We can further improve the efficiency by setting a probability
threshold during generation. That is, we only generate tuples
$(r_i,r_j,s,p_s)$ whose $p_s > \rho$. In other words, when
generating tuples we only consider subrecords whose probability of
being a signature exceeds the threshold $\rho$. This filtering allows
us to remove a large number of subrecords with low probability of
being signatures at an early stage.

The Elimination step above can be skipped, if the precomputed
subrecords from each raw record by design do not contain each other
as subrecords.

After obtaining the probability for a pair of records to share a
signature, we can place the two in a block if this probability exceeds the threshold $\tau$. Note that blocks built this way contain two and only two records each.
One can then employ any similarity function, such as Jaccard similarity, edit distances like Levenhstein and Jaro, or some other domain-specific functions~\cite{Christen:2012:DMC:2344108}, to decide whether to link them at all.
When the parameter $a$, $b$, and $\tau$ are carefully tuned using training data, one can simply link all pairs of records sharing a probability higher than threshold $\tau$.


\section{Connected Components: A Scalable In-Database Algorithm}\label{sec:cc}

The previous section describes how pairs of records can be linked via
probabilistic identification of common signatures.
In this section, we present a scalable algorithm to assign a
consistent label (entity identifier) to records which are linked
either directly or indirectly. The problem is equivalent to the
problem of finding connected components in a general graph \cite{Cormen:2009}, except that the graph in our case is too large to allow random access.
In the following, we propose a
connected-component labelling algorithm that works on large graphs
stored in a distributed, parallel database.

Without loss of generality, we will label each connected component
with the smallest node (record) identifier of the component. Our
algorithm contains two iterative steps. We first transform the input
graph into equivalent trees (a forest) such that nodes on each
connected component are in the same tree, and that the identifier of
a descendant is always larger than that of its ancestors. We then
transform the forest into an equivalent forest in which the height of
all the trees is one. Upon convergence, all nodes in the same
connected component will be connected directly to the root node,
which can then be used as the consistent identifier for all entities
in the tree.

Figure~\ref{fig:cc} shows an example. The input (left) is a set of
node-pairs ($e_1$,$e_2$), ($e_1$,$e_4$), ($e_2$,$e_3$), ($e_2$,$e_4$),
($e_2$,$e_5$), and ($e_3$,$e_5$). Without losing generality, we
always use the smaller entity identifier as the first element in
each pair. We know this is not yet a forest because some nodes, such
as nodes $e_4$ and $e_5$, have more than one parents. When a node has
more than one parents, namely when the node-pairs contain patterns
like ($e_1$,$e_j$), ($e_2$,$e_j$), $\dots$, and ($e_i$,$e_j$), we do
the following replacement:
\begin{align}\label{eq:replacement 1}
(e_1,e_j),(e_2,e_j),\dots,(e_i,e_j) &\Rightarrow \nonumber\\
(e^\star,e_j),(e^\star,e_1),&(e^\star,e_2),\dots,(e^\star,e_i)
\end{align}
where 
\begin{equation}
e^\star=\min(e_1,e_2,\dots,e_i)\quad.
\end{equation}
This is a grouping operation per node $e_j$ that can be implemented
efficiently in a parallel database. During the replacement we drop
duplicated edges and self-loops (an edge connecting a node to the
node itself).

\begin{figure}[t!]
\centering
\includegraphics[width=0.45\textwidth]{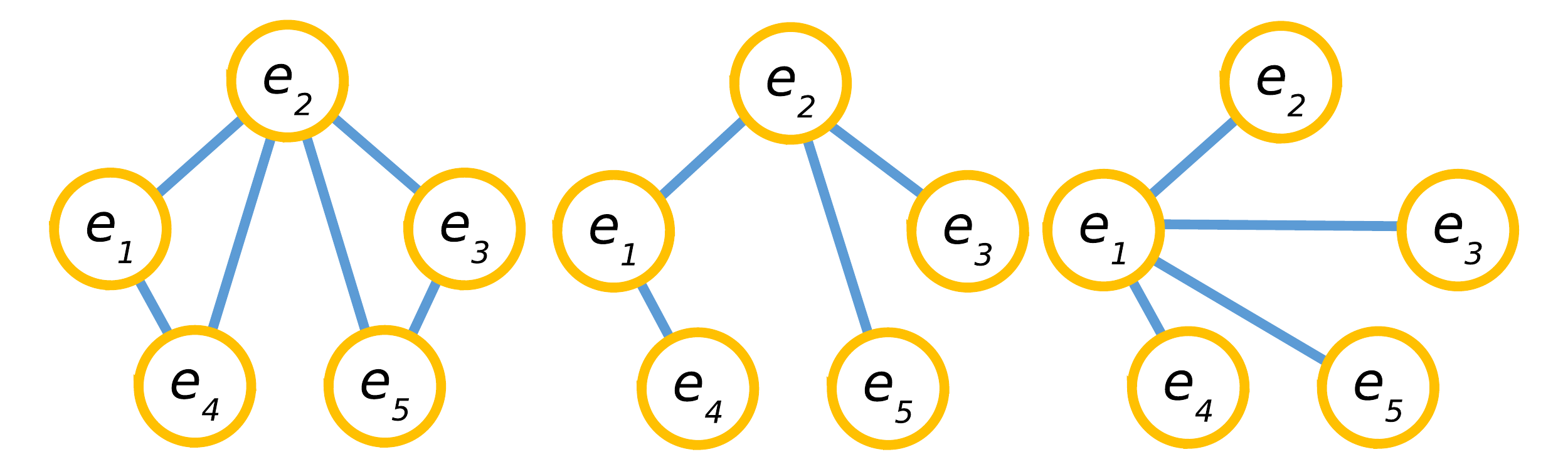}
\caption{\label{fig:cc}Left: input graph; middle: transformed to
trees; right: reduce tree height to one.}
\end{figure}

Through such a replacement, we guarantee that
\begin{enumerate}\itemsep0mm\parskip0mm
  \item the connections between $e_1$, $e_2$, $\dots$, $e_i$, $e_j$
        are preserved; and
  \item $e_j$ ends up with a single parent.
\end{enumerate}
The newly added node pairs may introduce new parents to existing
nodes in the graph. We therefore apply the replacement step
(Equation~(\ref{eq:replacement 1})) recursively until every node has a
single parent. Convergence is guaranteed because the sum of node
identifiers in the list is non-negative and each replacement always
reduce this sum by a positive integer. Upon convergence of the first
replacement step, we obtain the second graph (middle) in Figure~\ref{fig:cc}
which is a forest with node-pairs ($e_1$,$e_2$), ($e_1$,$e_4$),
($e_2$,$e_3$), and ($e_2$,$e_5$).

A tree's height is larger than one if its node-pairs contain
patterns like  ($e_i$,$e_j$) and ($e_j$,$e_k$), namely a node exists
as a parent and a child at the same time. For a tree whose height is
larger than one, we iteratively do the following replacement 
\begin{equation}
(e_i,e_j),(e_j,e_k)\Rightarrow(e_i,e_j),(e_i,e_k)
\end{equation}
until the height of all trees become one, as shown in the right side of Figure~\ref{fig:cc}. This is a join operation
that can be implemented efficiently in a parallel database. If we
denote by $h$ the height of the highest tree in the forest, then
the above 
converges in $\log_2(h)$ rounds.


\section{The p-Signature Algorithm}
\label{sec:signature-er}

We are now ready to present our signature-based algorithm for ER, which is given in Algorithm~\ref{alg:1}.
The algorithm requires these inputs:
\begin{itemize}
 \item $O=\{r\}$, a set of unlabelled observations;
 \item $U=\{s\}$, a set of subrecords selected by users as candidate signatures based on domain knowledge;
  \item $\rho$ and $\tau$, thresholds: we consider a
subrecord only if its probability of being a signature
exceeds $\rho$; and we adopt a link if the probability for two
records to share a signature exceeds $\tau$; and
 \item $v$, an optional similarity function.
\end{itemize}

The first four steps of the algorithm are as described in Section~\ref{subsec:record linkage}.
In the algorithm, $\gets$ denotes the operation of adding an element to a set and $\backslash$ denotes the operation of removing an element from a set.

In Step~\ref{algo1:step1}, $I=\{(s,R_s,p_s)\}$ denotes the inverted index of $O$ with respect to $U$, where $s\in U$, $R_s\subseteq O$ denotes the set of records all containing $s$ as a subrecord, and each $p_s=P(s\in S\mid k=|R_s|)$ is the probability of $s$ being a signature given that $s$ appears in $|R_s|$ different records in the database (see Equation~(\ref{eq:result})). 
In Step~\ref{algo1:step2}, the condition $r_i < r_j$ is there to ensure we don't generate symmetric entries.
Step~\ref{algo1:step3} can be done because of Proposition~\ref{prop:superrecord signature}.
Step~\ref{algo1:pairwise} selects the final pairwise linkages based on the potential linkages computed earlier. 
The first three steps can be thought of as the blocking/indexing step in a standard ER framework, and Step~\ref{algo1:pairwise} can be thought of as the record comparison step. 
At the end of Step~\ref{algo1:pairwise}, $L=\{(r_i,r_j)\}$
holds all the detected links between records in $O$.
In Step~\ref{algo1:last step}, $c$ denotes the connected components algorithm described in Section~\ref{sec:cc}.

\begin{figure}[t!]
\begin{algorithm}[H]
\begin{algorithmic}[1]
\REQUIRE $O=\{r\}$, $U=\{s\}$, $\rho$, $\tau$, $v$ \medskip
\STATE\label{algo1:step1} Build inverted index: 
  $$I \gets (s,R_s,p_s) ~,\; \forall s \in U, R_s \subseteq O, p_s
  > \rho$$
\STATE\label{algo1:step2} Generate potential linkages:
  $$K \gets (r_i,r_j,s,p_s)~, \forall (s,R_s,p_s) \in I,~ r_i,r_j
    \in R_s, r_i<r_j$$ 
\STATE\label{algo1:step3} Eliminate redundant linkages: 
  $$K=K \backslash (r_i,r_j,s,p_s),~ \forall r_i, r_j, s,p_s $$
  if $(r_i,r_j,\hat{s},p_{\hat{s}})\in K$ and
  $s$ is a subrecord of $\hat{s}$.\medskip
  
\STATE\label{algo1:pairwise} Finalise pairwise linkages:
  $$L\gets(r_i,r_j)$$
  for all $r_i$ and $r_j$ such that
  $$ 1-\prod_{(r_i,r_j,s,p_s)\in K}(1-p_s) >\tau\quad $$
  and $v(r_i,r_j)=\text{true}$ \medskip

\RETURN $c(L)$.\label{algo1:last step}
\end{algorithmic}
\caption{\label{alg:1} Signature-based Entity Resolution}
\end{algorithm}
\end{figure}


\subsubsection*{Candidate Signatures}

The ER algorithm above requires a user to specify a set of
candidate signatures as input. 
These candidate signatures have an impact on both the accuracy and computational complexity of the algorithm and should be chosen based on domain
knowledge about a database and can differ from case to case. 
In Section~\ref{sec-experiments}, we will provide some concrete examples of candidate signature specifications and discuss the issue of how to construct good candidate signatures.

\subsubsection*{Post-Verification Rules}
An important but optional parameter in Algorithm~\ref{alg:1} is $v$, the post-verification rules.
It is largely an optional parameter when training data is available to tune the other parameters.
But when training data is not available, $v$ is a mechanism for the user to supply additional domain knowledge to improve ER accuracy.
The post-verification rules can be as simple as a suitably thresholded distance function like Jaccard or Jaro-Winkler.
However, it is more commonly used to resolve prickly and context-dependent cases like family members that live in the same address, a person and his company (e.g. John Smith and John Smith Pty Ltd), and distinct franchisees that use a common bank account.

\subsubsection*{Computational Complexity and Implementation}
The computational complexity of Algorithm~\ref{alg:1} is dominated by the first two steps, which have time and space complexity $O(m)$, where $m$ is the number of distinct candidate signatures extracted.
Most natural choices of candidate signatures leads to $m \sim O(n)$, where $n$ is the size of the (deduplicated) dataset. The scalability of the algorithm is studied empirically in Section~\ref{sec-experiments}.

We have two implementations of the algorithm, one in SQL running on Greenplum, and one in Scala running on Spark.
The SQL code is similar in structure to that in \cite{zhang17} and
involves only joins (all efficiently executable using hash-join \cite{zeller90}) and straightforward group-by operations.
The Spark version has less than 100 lines of code and is the one we use in a production system.
Both the parallelised SQL and Spark code are undergoing due dilligence to be made open-source and available on Github.


\section{Experimental Evaluation}
\label{sec-experiments}

\newcolumntype{C}[1]{>{\centering\let\newline\\\arraybackslash
  \hspace{0pt}}m{#1}}
\newcolumntype{M}[1]{>{\centering\let\newline\\\arraybackslash
  \hspace{0pt}\columncolor[gray]{0.9}}m{#1}}

\begin{table*}[th!]
  \caption{\label{tab:statistics}A summary of the results of applying the proposed
    algorithm on benchmark datasets.}
  \centering
  \begin{small}
  \begin{tabular}{|c|c|c|c|c|c|c|c|c|c|c|} \hline
    ~ & DBLP & Scholar & DBLP & ACM & Abt & Buy & Amazon &
      Google & NCVR-2014 & NCVR-2017 \\ \hline \hline
  Records & 2,616 & 64,263 & 2,616 & 2,294 & 1,081 & 1,092
    & 1,363 & 3,226 & 5,616,368 & 7,861,249 \\ \hline
  Subrecords & 547,722 & 6,052,597 & 742,952 & 558,731 &
    24,348 & 25,179 & 21,037 & 77,787 & 131,218,277 &
    162,115,747 \\ \hline
  Ground truth & \multicolumn{2}{c|}{5,347} & \multicolumn{2}
    {c|}{2,224} & \multicolumn{2}{c|}{1,097} & \multicolumn{2}
    {c|}{1,300} & \multicolumn{2}{c|}{5,015,915} \\\hline 
  Precision & \multicolumn{2}{c|}{91.0\%} & \multicolumn{2}
    {c|}{97.7\%} & \multicolumn{2}{c|}{87.9\%} & \multicolumn{2}
    {c|}{60.2\%} & \multicolumn{2}{c|}{96.3\%} \\ \hline
  Recall & \multicolumn{2}{c|}{89.5\%} & \multicolumn{2}{c|}{97.4\%}
    & \multicolumn{2}{c|}{60.4\%} & \multicolumn{2}{c|}{66.1\%} &
    \multicolumn{2}{c|}{89.5\%}\\ \hline
  F-measure & \multicolumn{2}{c|}{90.2\%} & \multicolumn{2}
    {c|}{97.6\%} & \multicolumn{2}{c|}{71.6\%} & \multicolumn{2}
    {c|}{63.0\%} & \multicolumn{2}{c|}{92.8\%} \\ \hline
  Time & \multicolumn{2}{c|}{10 sec} & \multicolumn{2}{c|}{6 sec} &
    \multicolumn{2}{c|}{6 sec} & \multicolumn{2}{c|}{10 sec} &
    \multicolumn{2}{c|}{307 sec}\\\hline  
	\end{tabular}
  \end{small}
\end{table*}

We use six different ER problems to empirically evaluate the proposed algorithm.
The entities in these six
problems range from academic publications and commercial products, to
individuals and organisations. The datasets range from thousands
to billions of records in size. There is also a large
diversity of data quality issues, including incompleteness,
incompatible formats, errors, and temporal inconsistency. We use
these datasets to benchmark the accuracy as well as scalability of
our proposed algorithm. 
All the experiments are done using the open-source Greenplum Database running on 8 servers (1 master + 7 slaves), each with 20 cores, 320 GB, and 4.5 TB usable RAID10 space.
The results are summarised in Table~\ref{tab:statistics}.


\subsection{Entity Resolution Quality}

In the first experiment, we apply our algorithm to the four publicly
available datasets evaluated in~\cite{Kopcke:2010:EER:1920841.1920904} 
where ground truth is available:
(1) DBLP-ACM, (2) DBLP-Google Scholar, (3) Apt-Buy, and (4)
Amazon-Google-Products.
The entities in the first two datasets are academic publications,
and each record contains title, authors, venue, and year of 
publication. The entities in the third and fourth datasets are
consumer products, and each record contains name, description,
manufacturer, and price.

\begin{table*}[th!]
  \caption{\label{tab:f} F-measure of our proposed method
      ($\sqsubseteq$-signature) as well as five existing methods on
      four benchmark datasets, as reported
      by~\protect\cite{Kopcke:2010:EER:1920841.1920904}. The top
      performer of each dataset is presented in bold font.}
  \centering
  \begin{small}
  \begin{tabular}{|l|C{0.1\textwidth}|C{0.1\textwidth}|C{0.1
    \textwidth}|C{0.1\textwidth}|C{0.1\textwidth}|C{0.1\textwidth}|}
      \hline
    ~ & COZY & FEBRL \newline FellegiSunter & FEBRL SVM & MARLIN
      ADTree & MARLIN SVM & $\sqsubseteq$-Signature\\ \hline \hline
      DBLP-Scholar & 82.9 & 81.9  & 87.6 & 82.9 & 89.4 & \textbf{90.2}
      \\ \hline
    DBLP-ACM & 93.8 & 96.2 & \textbf{97.6} & 96.4 & 97.4 &
      \textbf{97.6} \\ \hline
    Abt-Buy & 65.8 & 36.7 & 71.3 & 54.8 & 70.8 &\textbf{71.6} \\
      \hline
    Amazon-Google & 62.2 & 53.8 & 60.1 & 50.5 & 59.9 &
      \textbf{63.0}\\\hline
  \end{tabular}
  \end{small}
\end{table*}

For academia publications, we use the following two types of subrecords as candidate signatures:
\begin{enumerate}\itemsep1mm\parskip0mm
\item three consecutive words in title; and
\item two consecutive words in title, plus two random words in
      authors.
\end{enumerate}
For commercial products, we use the following three types of candidate signatures:
\begin{enumerate}\itemsep1mm\parskip0mm
\item one word from name;
\item two consecutive words from name; and
\item three consecutive words from name.
\end{enumerate}

Following previous evaluation 
work~\cite{Kopcke:2010:EER:1920841.1920904}, we run our algorithm
multiple times with varying parameters 
and then pick the best-performing model. We use F-measure to quantify the
performance, which is defined as the harmonic mean of precision and recall
\cite{Christen:2012:DMC:2344108}.
We note that the legitimacy of using the F-measure to evaluate ER algorithms is questioned in a recent paper~\cite{Hand2017}.
However, we use the F-measure here because it allows direct
comparisons with the earlier evaluation presented
in~\cite{Kopcke:2010:EER:1920841.1920904} (which does not include
precision and recall results).

The result of our method and five other algorithms, three of which
are supervised machine learning based classification algorithms, are
presented in Table~\ref{tab:f}. The performance of the other
algorithms is taken from \cite{Kopcke:2010:EER:1920841.1920904}.
The top performer for each dataset is highlighted in bold.
Our proposed method turns out to achieve state-of-the-art results 
on all four datasets (tied for first in one case).
Although the winning margin may not always be statistically significant, the consistent
good performance across the four diverse datasets is noteworthy, however.

Note also that the performance of our method is achieved by fixed
subrecord types as described above. It is possible to further
improve the current performance with other types of subrecords that
are customised for each dataset. 


\subsection{Entity Resolution Scalability}

To test the scalability of our method, we employ it to link records across
two snapshots of the North Carolina Voter Registration (NCVR)
database (\url{http://dl.ncsbe.gov/}). We used a snapshot from
October 2014 and linked it with a snapshot from October 2017. We
used the following information of each voter for the linkage:
\begin{itemize} \itemsep1mm\parskip0mm
	\item full name (first, middle, and last name);
	\item residential address (street, city, zipcode and state);
	\item mail address (street, city, zipcode and state);
	\item phone number;
	\item birth state; and
	\item birth age.
\end{itemize}
Note that there is a temporal aspect to this particular ER problem,
in that each attribute above for the same voter may change over 
the three years, except birth state and birth age (with age being
increased by three from 2014 to 2017). Among the 5,015,915 voters
who appear in both datasets, the percentage of voters who changed
their name, residential address, mail address, or phone number are
$5\%$, $33\%$, $33\%$, and $48\%$, respectively. Moreover, $3\%$ of
the voters changed their birth state, and $6\%$ of the voters have
inconsistent age (not differing by 3 years) in the two datasets.
Each voter also has an identifier (NCID), which is used to generate
the ground truth for ER.

We used the following subrecords as candidate signatures:
\begin{itemize} \itemsep1mm\parskip0mm
  \item two random words from name, two consecutive words from
        residential address;
  \item two random words from name, two consecutive words from mail
        address;
  \item two random words from name, last six digits from phone number;
        and
  \item full name, birth state, birth age;
\end{itemize}
where birth age from NCVR-2014 is incremented by 3 to align with
that in NCVR-2017.

As Table~\ref{tab:statistics} shows, 
while the size of the NCVR dataset is
about 1,000 times larger than the other benchmark datasets, 
the total time used for ER only increased 30 to 50 times.

No previous ER work has been applied to the same NCVR dataset at the scale we have done, which makes comparison difficult.
A relevant previous work is \cite{Hu2017}, which randomly sampled subsets of
size 5,000, 10,000, 50,000, and 100,000 from the NCVR dataset to
implement temporal ER. As \cite{Hu2017} shows, the performance of
the considered algorithms monotonically declines as the size of the
sampled dataset increases. The top performer on the largest subset,
which contains 100,000 records, achieved an F-measure of $92\%$ per
\cite{Hu2017}. Our method is applied to the complete datasets
between two time points and achieved comparable accuracy.


\subsection{Large Scale Transitive ER}

So far we only considered the scenarios where ER is between two
different datasets in a pairwise manner. Now we consider ER
within a single dataset (deduplication). The considered dataset is maintained by
an Australian Government agency, containing over 690 million reports
submitted by over 10,000 organisations over 10 years. More than 3.9
billion individuals and organisations appear in these
reports. Our aim is to identify records by the same individuals and
organisations and link them together.

When an entity appears in a report, some or all of the following
information may be provided: name, proof of ID (such as driver's
license or passport), address, date of birth, company number,
telephone number, email, bank account number.
%
The format of each type of information differs from report to report.
In most reports, one or more attributes are not available. Since we
have no ground truth for this dataset, we report only the
scalability of our proposed algorithm.

After removing exact duplicate records, the number of distinct
records was reduced to around 300 millions. To handle the poor
data quality, we generated 13 types candidate signatures from each record. In particular, the first 7 types of candidate signatures contain two random words from name followed by any of the following
\begin{enumerate}\itemsep1mm\parskip0mm
  \item two consecutive address words;
  \item last six digits of ID number;
  \item date of birth;
  \item last six digits of company
    number;
  \item last six digits of telephone,
    number;
  \item  email; and
  \item last six digits of account number.
\end{enumerate}
The other 6 types of candidate signatures contain two consecutive address words followed by either of item 2-7 above.

We do not require two name words to be consecutive to allow names in
inconsistent formats to be compared. We however require address words
to maintain their input order because the order of address words is
more consistent than that of name, and an address is usually much
longer than a name, and there would be too many unordered
combinations to consider. We use the last six digits of account
number, telephone number, and proof of ID, because these attributes
are usually longer than six digits, the ending parts of these
attributes usually have more consistent format than their starting
parts, and being identical in the last six digits rarely leads to
false matches especially when they are concatenated with name.

\begin{table} 
\caption{\label{tab:huge} Large-scale Transitive ER: size of each intermediate output and the time taken} 
  \centering
  \begin{small}
  \begin{tabular}{|l|r|r|} \hline
  & Size~~~~~~ & Time~~~ \\ \hline \hline
  Records & 3,989,630,008 & \\ \hline
  Distinct records & 268,657,406 & 1,585 sec \\ \hline
  Candidate signatures & 4,084,438,114 & 626 sec \\ \hline
  Pairwise links & 1,002,675,163 & 6,839 sec \\ \hline
  Verified links & 623,498,453 & 3,083 sec \\ \hline
  Connected components & 148,163,665 & 496 sec \\ \hline
  Overall on Greenplum & & 12,629 sec \\ \hline\hline
  Overall on SparkSQL && 5,044 sec\\ \hline
\end{tabular}
\end{small}
\end{table}

One practical difficulty in applying the proposed algorithm to a real and large
dataset is that we have no labelled data to tune our parameters.
In our business context, a false link usually has a much higher cost than a
missing link. 
We therefore adopted some post-verification rules such as Jaccard distance on
linked entities to further improve precision at the cost of lower recall.

Some statistics of our proposed method on this large dataset is
given in Table \ref{tab:huge}. As can be seen, resolving over 3.9
billion records with the proposed method takes around three and a
half hours. Compared to resolving 12 million records in the NCVR
datasets in 307 seconds, our algorithm scales in sublinear time.

Besides Greenplum, we also implement our algorithm with SparkSQL and
resolve the over 3.9 billion entities a server which 4-time as large as the Greenplum server. 
The processing time reduces to 5,044 seconds. Note that the 5,044 seconds include the time of saving output of each step to HDFS for debugging purpose.

\subsection{Practical Considerations}

We now discuss several important practical considerations of our approach.
\medskip

\noindent {\bf Choice of Candidate Signatures:}
As stated earlier, the choice of candidate signatures depends on domain knowledge and has an impact on both the accuracy and computational complexity of the ER algorithm.
Here are some general guidelines on setting this parameter.
\begin{enumerate}\itemsep1mm\parskip0mm
  \item A candidate signature should be short so that it has a good
        chance of recurring in multiple records. 
  \item A candidate signature should be distinctive enough so that
        it has a good chance to be a signature.
  \item All (unambiguous) records should have at least one non-empty signature.
\end{enumerate}

The three guidelines can pull us in opposite directions. 
As can be seen in Section~\ref{sec-experiments}, we usually want to extract small subrecords from key attributes in a record as candidate signatures, but these subrecords may not be sufficiently distinctive on their own.
An effective way to improve the distinctive power of such short candidate signatures
is to concatenate subrecords from multiple attributes, such as using
\emph{name}+\emph{address}, \emph{name}+\emph{phone number}, and so on. 

To the extent possible, we want to make sure each record in the dataset has at least one signature 
by making sure at least one candidate signature with sufficiently high probability can be extracted from each record. 
This is not always possible when there exist inherently ambiguous records like (John, Sydney NSW) that cannot be adequately resolved no matter what.
But there are plenty of interesting cases in the spectrum of distinctiveness that we would need to handle. Examples of difficult cases include names like John James Duncan (all common first names), names from certain ethnicity like Arabic and Vietnamese names, and addresses in certain countries like India.
In such situations, we should take longer candidate signatures into consideration.

When prior knowledge is not available or inadequate, we can generate candidate
signatures randomly. Because of our probabilistic formulation,
randomly generated subrecords are unlikely to cause false links 
but to fully link all relevant records, we may need
to generate a large number of candidate signatures.
In such cases, we may resort to the use of grammars \cite{cohen94,lloyd03logic-learning} to concisely define a search space of candidate signatures that can be  enumerated in a systematic and exhaustive way for testing.

The sensitivity of our ER algorithm to the choice of candidate signatures is both a strength and a weakness.
It is a strength in that when good domain knowledge is present, the candidate signatures provide a natural mechanism to capture and exploit that domain knowledge to achieve good accuracy and scalability. 
Many existing ER algorithms do not have such a natural mechanism to exploit available domain knowledge.
On the other hand, the sensitivity to the choice of candidate signatures is obviously a weakness in ER applications where no good domain knowledge is available, in which case other ``parameter-free" algorithms may be more suitable. \medskip

\noindent {\bf Efficiency Overkill?}
Do we really need an ER algorithm that can process millions of records in a few hours? 
Ideally, data volume at that scale are processed once using a batch algorithm and then an incremental algorithm is used to incorporate new data points as they appear.
In practice, many ER algorithms do not have an incremental version.
Even when they do, the results obtained from the batch and incremental algorithms are usually not perfectly consistent with each other.
In our actual target application, up to 1 million new records are added to the database every day.
Incrementally resolving such large numbers of new records in a manner that maintains consistency with the batch algorithm -- a key requirement in the intelligence domain where analytical results can be used as evidence in court proceedings -- is as hard as the problem of performing batch ER on the new full dataset.
Having an ER algorithm that can be rerun on the full dataset every day in a couple of hours is thus important in our setting.
Further, such an efficient algorithm gives us the agility to make changes and experiment with parameters during the development of the algorithm, something impossible to do if the algorithm take days or weeks to process the entire dataset. \medskip


\noindent {\bf Limitations of P-Signature:}
For efficiency, we choose not to compute all the common subrecords between a pair of records, but to approximate them with a set of precomputed subrecords, typically of limited length. When the precomputed subrecords of a record are all non-distinctive, we will not be able to link this record distinctively to other records of the same entity. To improve the situation, one may consider more diversified and longer candidate signatures at the price of lower efficiency. Besides, the granularity of our token set $W$ also affects how robust our signatures are against inconsistency. Currently words are the finest granularity of our algorithm. That means, we will not be able to link a record if it contains typos in every word. To tackle this challenge, we need to define our vocabulary on q-grams (character substrings of length $q$) or even individual characters instead. Yet in return, the distinctiveness of each candidate-signature will be weaker. The challenge is, following its current design, P-Signature can hardly link Smith with Smithh, but not link Julie with Juliet at the same time.


\section{Related Work and Discussion}\label{sec:rel-work}
We will start with a review of related work followed by a discussion of key connections between our signature ER framework and some existing ER techniques.

\subsubsection*{Related Work in Entity Resolution}
Entity resolution (ER), also known as record linkage and data
matching~\cite{Christen:2012:DMC:2344108}, has a long history with
first computer based techniques being developed over five decades
ago~\cite{fellegi-sunter69,New59}. The major challenges of linkage
quality and scalability have been ongoing as databases continue to
grow in size and complexity, and more diverse databases have to be
linked~\cite{Don15}. ER is a topic of research in a variety of
domains, ranging from computer
science~\cite{Christen:2012:DMC:2344108,Don15,Nau10} and
statistics~\cite{Her07} to the health and social
sciences~\cite{Har15}. While traditionally ER has been applied on
relational databases, more recently the resolution of entities in
Web data~\cite{DBLP:series/synthesis/2015Christophides} has become
an important topic where the aim is to for example facilitate
entity-centric search. The lack of well defined schemas and data
heterogeneity~\cite{Has13}, as well as dynamic data and the sheer size of Web
data, are challenging traditional ER approaches in this
domain~\cite{DBLP:series/synthesis/2015Christophides}.

The general ER process can be viewed to consist of three major
steps~\cite{Christen:2012:DMC:2344108}: blocking/indexing, record
comparison, and classification, which is sometimes followed by a
merging step \cite{Ben09,DBLP:series/synthesis/2015Christophides}
where the records identified to refer to the same entity are
combined into a new, consistent, single record.

In the first step, as discussed earlier, the databases are split
into blocks (or clusters), and in the second step pairs of records
within the same blocks are compared with each other. Even after data
cleansing and standardisation of the
input databases (if applied) there can still be variations of and
errors in the attribute values to be compared, and therefore 
approximate string comparison functions (such as edit distance, the
Jaro-Winkler comparator, or Jaccard 
similarity~\cite{Christen:2012:DMC:2344108,Nau10}) are employed to
compare pairs of records. Each compared record pair results in a
vector of similarities (one similarity per attribute compared), and
these similarity vectors are then used to classify record pairs into
\emph{matches} (where it is assumed the two records in a pair
correspond to the same entity) and \emph{non-matches} (where the
records are assumed to correspond to two different entities).
Various classification methods have been employed in
ER~\cite{Christen:2012:DMC:2344108,Don15,Nau10}, ranging from
simple threshold-based to sophisticated clustering and supervised
classification techniques, as well as active learning approaches.

Traditional blocking~\cite{Chr12b} uses one or more attributes as
\emph{blocking key} to insert records that share the same value in
their blocking key into the same block. Only records within the same
block are then compared with each other. To overcome variations and
misspellings, the attribute values used in blocking keys are often
phonetically encoded using functions such as Soundex or
Double-Metaphone~\cite{Christen:2012:DMC:2344108} which convert a
string into a code according to how the string is pronounced. The
same code is assigned to similar sounding names (such as `Dickson'
and `Dixon'). Multiple blocking keys may also be used to deal with
the problem of missing attribute values~\cite{Chr12b}.

An alternative to traditional blocking is the sorted neighbourhood
method~\cite{Dra12,Mon96,Nau10}, where the databases to be
linked are sorted according to a \emph{sorting key} (usually a
concatenation of the values from several attributes), and a sliding
window is moved over the databases. Only records within the window
are then compared with each other. Another way to block databases is
using canopy clustering~\cite{Mcc00a}, where a computationally
efficient similarity measure (such as Jaccard similarity based on
character q-grams as generated from attribute values~\cite{Nau10})
is used to inserts records into one or more overlapping clusters,
and records that are in the same cluster (block) are then compared
with each other.

While these existing blocking techniques are \emph{schema-based} and
require a user to decide which attributes(s) to use for blocking,
sorting or clustering, more recent work has investigated
\emph{schema-agnostic} approaches that generate some form of
signature for each record automatically from all attribute
values~\cite{DBLP:series/synthesis/2015Christophides, papadakis2015schema,papadakis13tkde,War06}. While schema
agnostic approaches can be attractive as they do not
require manual selection of blocking or sorting keys by domain
experts, they can lead to sub-optimal blocking performance and might
require additional meta-blocking
steps~\cite{DBLP:series/synthesis/2015Christophides,
Efthymiou:2017:PMS:3050918.3050949,papadakis14tkde} to
achieve both high effectiveness and efficiency by for example
removing blocks that are too large or that have a high overlap with
other blocks.

One schema-agnostic approach to blocking is Locality Sensitive
Hashing (LSH), as originally developed for efficient
nearest-neighbour search in high-dimensional spaces~\cite{Ind98}.
LSH has been employed for blocking in ER by hashing attribute
values multiple times and comparing records that share some hash
values. One ER approach based on MinHash~\cite{Bro1997} and LSH is
\emph{HARRA}~\cite{Kim10}, which iteratively blocks, compares, and
then merges records, where merged records are re-hashed to improve
the overall ER quality. 
However, a recent evaluation of blocking
techniques has found~\cite{Ste14}, blocking based on LSH needs to be
carefully tuned to a specific database in order to achieve both high
effectiveness and efficiency. This requires high quality training
data which is not available in many real-world ER applications.


With the increasing sizes of databases to be linked, there have been
various efforts to parallelize ER algorithms, where both
algorithmic~\cite{Kaw06,Kim07} as well as platform dependent
(assuming for example Map Reduce)~\cite{Efthymiou:2017:PMS:3050918.3050949,Kol13} solutions
have been proposed. A major challenge for parallel ER is load
balancing due to the irregular distribution of the data, resulting
for example in blocks of very different sizes.

Compared to existing approaches to ER, the distinguishing feature of our ER algorithm is a data-driven blocking-by-signature technique that deliberately trade-off recall in favour of high precision. 
This is in contrast to the standard practice of trading off precision in favour of high recall in most existing blocking algorithms.
To compensate for potential low-recall resulting from our blocking technique, we introduce an additional Global Connected Component step into the ER process, which turns out to be efficiently computable.
As shown in Section~\ref{sec-experiments}, this slightly unusual combination of ideas yielded a new, simple algorithm that achieves state-of-the-art ER results on a range of datasets, both in terms of accuracy and scalability.

\subsubsection*{Connections to the Signature ER Framework}
Perhaps unsurprisingly, many existing ER techniques can be understood in the signature framework described in Section~\ref{sec:problem formulation}. We now point out a few such connections. \medskip

\noindent {\bf Standardisation and Cleansing:} 
The most common question on our ER algorithm from industry practitioners is the (deliberate) avoidance of an explicit up-front data standardisation and cleansing step. 
We now address this.
The canonical form of each record obtained from standardisation and cleansing is actually a type of signature. 
Whereas the transform from a record to signatures is the generation of subrecords in our algorithm, in traditional methods the transforms are context- and data-dependent and usually implemented using business rules that can become complicated and hard-to-maintain over time. 
The main benefit of using standardisation and cleansing transforms is that the derived canonical form is almost guaranteed to be a signature. 

In our method, by contrast, a derived subrecord is only a signature
with a certain probability.
To compensate, we generate many subrecords for each
database record. 
An important benefit of generating many subrecords or signatures
is that two records will be linked if any of these signatures are
shared. In contrast, data standardisation methods produce only one
signature from each record, and the signature/canonical form for
two records of the same entity may be quite different. 
This issue is then (partially) addressed by allowing the signatures to be matched in a non-exact way using similarity measures that capture different criteria.

To summarise, our method generates low-cost signatures and match signatures exactly. 
We put uncertainty in whether a generated subrecord is a signature and mitigate the
risk with a number of candidate signatures. 
Traditional ER methods that rely on data-standardisation generate high-cost signatures and match signatures approximately.
They put uncertainty in whether two signatures match or not.\medskip

\noindent {\bf MinHash and LSH:}
There are several connections between $\sqsubseteq$-signatures and the signatures generated by MinHash:
\begin{enumerate}
  \item Two records have an identical MinHash band (an array of
    MinHash values) if they both contain some words, and not contain
    some other words, at the same time. In our design, two records
    have an identical candidate $\sqsubseteq$-signature, as long as
    they both contain some words. 
	MinHash signatures
    therefore better fit scenarios where global consistency between
    two records matters, such as Jaccard similarity over large documents~\cite{Christen:2012:DMC:2344108}. Our $\sqsubseteq$-signatures better fits scenarios where partial similarity matters, for example where records of the same entity can contain significant inconsistency.
    
  \item Both method generate multiple signatures from each record. 
    Each MinHash band and candidate signature only captures partial (but important)    information in the original
    records. Therefore both methods allow inconsistent records to
    be linked together. 
    
  \item To achieve good balance between accuracy and efficiency,
    one can vary the length of signatures and the length
    of each band in MinHash. As shown in \cite{Ste14}, finding suitable
    values of these two parameters that lead to high quality ER
    results is context and data-dependent and requires ground truth data to tune. 
    In our method, the choice of candidate signatures
    and probability thresholds are parameters that can be similary tuned to achieve the same balance.
    
  \item Record linkage is probabilistic in both cases. MinHash has a probabilistic explanation with respect to the
    Jaccard similarity between two records~\cite{Bro1997}. Our
    method has a probabilistic explanation with respect to the
    co-occurrence of probable signatures in two records.
\end{enumerate}

\noindent {\bf Optimal Blocking:}
Our blocking approach is also related to the blocking framework described in \cite{papadakis14tkde}.
In particular, our method can be perceived as providing an approximate
way to construct ideal blocks. The blocks generated by our method
always contains only two records, and the number of blocks a record
may appear in is also upper-bounded. These criteria correspond to the
optimal Comparison Cardinality (CC) as discussed by Papadakis et
al.~\cite{papadakis13tkde}.
Papadakis et al. argue the optimal value for CC equals 2, which is when every block contains exactly two records, and each record appears in one and only one block. 
However, high CC is only a necessary but not sufficient condition for high-quality blocking. In practice, the higher CC is, the higher the risk is of missing a true match.
In our algorithm, blocks always contain only two records, but a record can belong to multiple blocks to minimise the risk of missing matches.

\section{Conclusion}\label{sec:conclusion}

We have presented and evaluated in this paper a novel Entity Resolution algorithm that
\begin{itemize}
 \item introduces a data-driven blocking and record-linkage technique based on the probabilistic identification of $\sqsubseteq$-signatures in data;
 \item incorporates an efficient connected-components algorithm to link records \emph{across} blocks;
 \item is scalable and robust against data-quality issues.
\end{itemize}
The simplicity and practicality of the algorithm allows it to be implemented simply on modern parallel databases and deployed easily in large-scale industrial applications, which we have done in the financial intelligence domain.


\balance

\bibliographystyle{abbrv}
\bibliography{mybib}

\end{document}